\documentclass{conm-p-l}

\usepackage{amssymb}
\usepackage{graphicx}

\copyrightinfo{2016}{American Mathematical Society}

\newtheorem{theorem}{Theorem}[section]
\newtheorem{lemma}[theorem]{Lemma}
\newtheorem{conjecture}[theorem]{Conjecture}

\theoremstyle{definition}
\newtheorem{definition}[theorem]{Definition}

\theoremstyle{remark}

\numberwithin{equation}{section}

\begin{document}

\title{How sticky is the chaos/order boundary?}

\author{Carl P. Dettmann}
\address{School of Mathematics, University of Bristol, Bristol BS8 1TW, UK}

\email{Carl.Dettmann@bristol.ac.uk}


\date{\today}

\begin{abstract}
In dynamical systems with divided phase space, the vicinity of the boundary between regular and chaotic regions is
often ``sticky,'' that is, trapping orbits from the chaotic region for long times.  Here, we investigate the stickiness
in the simplest mushroom billiard, which has a smooth such boundary, but surprisingly subtle behaviour.  As a measure
of stickiness, we investigate $P(t)$, the probability of remaining in the mushroom cap for at least time $t$ given uniform
initial conditions in the chaotic part of the cap.  The stickiness is sensitively dependent on the radius of the stem $r$ via
the Diophantine properties of $\rho=(2/\pi)\arccos r$.  Almost all $\rho$ give rise to families of marginally unstable
periodic orbits (MUPOs) where $P(t)\sim C/t$, dominating the stickiness of the boundary. Here we consider the case
where $\rho$ is MUPO-free and has continued fraction expansion with bounded partial quotients.  We show that $t^2 P(t)$
is bounded, varying infinitely often between values whose ratio is at least $32/27$.  When $\rho$ has an eventually
periodic continued fraction expansion, that is, a quadratic irrational, $t^2 P(t)$ converges to a log-periodic function.
In general, we expect less regular behaviour, with upper and lower exponents lying between 1 and 2.  The results may
shed light on the parameter dependence of boundary stickiness in annular billiards and generic area preserving maps.
\end{abstract}

\maketitle

\section{Introduction}
Understanding Hamiltonian dynamics with mixed regular and chaotic phase space remains one of the most
important and intractable problems in dynamical systems, with many open questions.
Much of the difficulty of such systems is that often (and probably ``typically'' in many senses)
the boundary between regular and chaotic regions (however defined) is fractal.  This is true even for
well studied and visualised classes of systems such as two dimensional billiards, in which a point particle
(of unit mass and speed, without loss of generality) moves uniformly, making mirror-like reflections with
the boundary of a domain $D\subset\mathbb{R}^2$~\cite{CM06,Tabachnikov05}.  Dynamics in any
billiard has one natural, ``equilibrium'' invariant measure: For the flow (continuous time) dynamics it is
uniform (that is, proportional to Lebesgue) in position in the domain and in the direction of motion.  For the
map (from one collision to the next) it is uniform in arc length and the tangential component of velocity.

\begin{figure}
\centerline{\includegraphics[width=400pt]{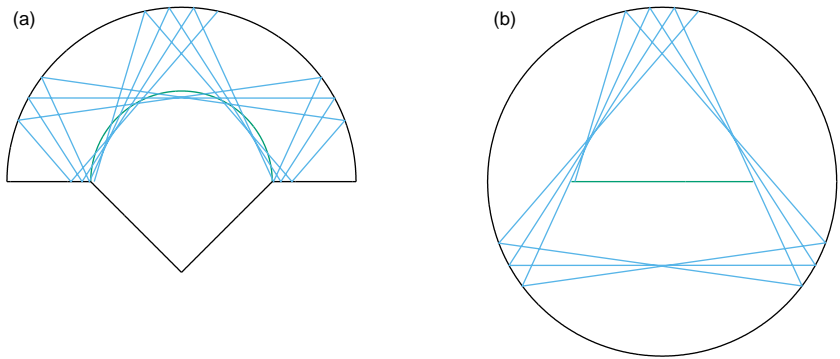}}
\vspace{-80pt}
\caption{(a) The mushroom geometry; the radius of the cap is $1$ and that of the stem is $r\in(0,1)$.
(b) Reduction of dynamics in the cap to a circular billiard, by reflecting across the $x$-axis.  The stem
becomes a slit-like hole of size $2r$ in the interior.  An orbit with rotation number close to $1/3$ is
shown in each case.}
\label{f:mush}
\end{figure}

Mushroom billiards were introduced by Bunimovich, as the first examples of billiards with sharply
divided regular and chaotic phase space~\cite{Bunimovich01}.  The simplest mushroom
consists of a semicircular cap (of unit radius, without loss of generality) and rectangular or triangular
symmetrically placed stem of radius $r\in(0,1)$; see Fig.~\ref{f:mush}.
Orbits with angular momentum (that is, distance of
closest approach to the centre) $l>r$ in the cap cannot reach the stem, and so form a regular
component with rotation number $\lambda/2$ where
\begin{equation}\label{e:lam}
\lambda=\frac{2}{\pi}\arccos l
\end{equation}
due to the integrability of the circle.  This region corresponds to $\lambda<\rho$ where
\begin{equation}\label{e:rho}
\rho=\frac{2}{\pi}\arccos r
\end{equation}
The factor $2$ in these definitions is needed to simplify the Diophantine approximation conditions;
see Ref.~\cite{DG11} and below.  The set of orbits that enter the stem forms the chaotic component, by
an application of the defocusing mechanism under which a large class of billiards with focusing boundary
components (including the well known stadium) can be shown to exhibit chaos~\cite{Bunimovich74}. In
fact, the mushroom limits to both the regular circle and ergodic semi-stadium in the limit of small and large
stem radius, respectively.  More complicated mushrooms may involve semi-ellipses and/or be constructed
to have an arbitrary number of regular and chaotic components~\cite{Bunimovich01,BV12}.

Later, it was observed that even in mushroom billiards, the simple structure of the phase space
may be complicated by the presence of marginally unstable periodic orbits (MUPOs) embedded in
the chaotic region~\cite{AMK05}.  As is common in the literature, we use the term MUPO to denote
the whole continuous family of periodic orbits with a given rational rotation number.  These are any periodic
orbit ($\lambda\in\mathbb{Q}$) restricted to cap of the mushroom (hence their marginal, ie parabolic nature)
but have $\lambda\geq\rho$ (hence located in the chaotic region or its boundary). When $\lambda$ is
perturbed, the orbit precesses for a long time but eventually falls
into the stem and into the main body of the chaotic region.  MUPOs are thus responsible for the phenomenon
of stickiness, the phenomenon in which chaotic orbits spend long periods in quasi-regular behaviour.  If the
mushroom stem has no periodic orbits entirely contained within it, such as the triangular stem in
Fig.~\ref{f:mush}, any MUPOs in the cap are the main source of stickiness. 

Note that there are a great variety of examples of stickiness in billiards, as discussed recently~\cite{BV12}.
``Internal'' stickiness is where there is no island of stability, such as the original stadium billiard, or where
stickiness is due to MUPOs completely contained
in the chaotic sea, such as the mushroom MUPOs above.   In contrast, ``external'' stickiness
is related to the boundary between chaotic and regular regions.  Ref.~\cite{BV12} gives
a number of examples of external stickiness, arguing that where the orbits in the regular
region are parabolic (as in mushrooms with circular caps) the island is typically not sticky, and
giving as a counterexample the case where the boundary corresponds to rational rotation
number (ie a MUPO). While MUPOs lead to stickiness located in the chaotic sea (internal stickiness),
the nature of mixed phase space also requires an understanding of the stickiness of the boundary between
the regular and chaotic regions (external stickiness).

Here we characterise stickiness in terms of an open billiard.  If initial conditions are distributed with respect to the
equilibrium invariant measure of the flow, restricted to the cap of a mushroom and to the chaotic region $\lambda>\rho$,
and the stem is replaced by a hole,  we can consider the survival probability $P(t)$, that the particle has not escaped
by time $t$.  A MUPO leads to a contribution to $P(t)$ proportional to $1/t$ as $t\to\infty$~\cite{DG11}.  That work
also gave a number of results (discussed below) about MUPOs in mushrooms, including characterising the set of $r$ for
which there are no MUPOs, including an explicit example, namely $r=\cos\left(\frac{5+\sqrt{2}}{23}\pi\right)\approx 0.640134$.

This paper is concerned with the stickiness of the boundary, which is weaker than stickiness due to MUPOs.  It
was conjectured in Ref.~\cite{DG11} that when there are no MUPOs the boundary would contribute $C/t^2$.
Subsequently Alastair Robertson's undergraduate project~\cite{Robertson13} with more careful numerics
for the above explicit $r$, suggested this form of decay, but with $C$ bounded but not convergent.
This calculation is replicated with a larger sample size of $10^{11}$ in Fig.\ref{f:mb64}.  

\begin{figure}
\centerline{\includegraphics[width=400pt]{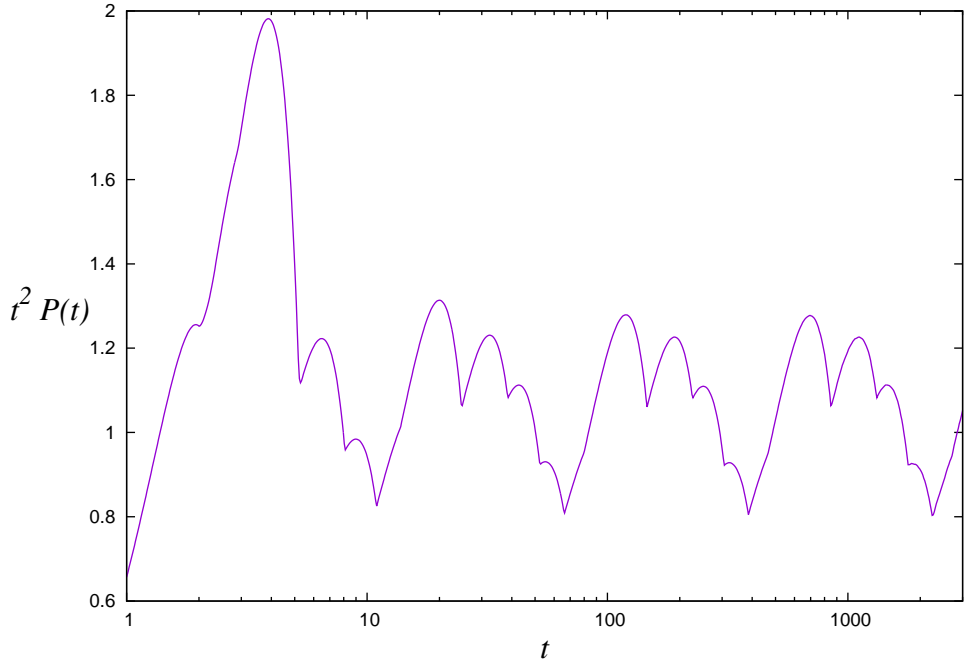}}
\vspace{-30pt}
\caption{In the MUPO-free case $r=\cos\left(\frac{5+\sqrt{2}}{23}\pi\right)\approx 0.640134$ the survival probability
limits to a log-periodic function divided by $t^2$, as shown by Thm.~\protect\ref{th:per}.}
\label{f:mb64}
\end{figure}

Here we find that not only the existence of MUPOs has an intricate parameter dependence, but also the boundary stickiness.
In particular, the stickiness depends sensitively on the Diophantine properties of $\rho$.
In this paper we give results when the partial quotients $a_n$ in the continued fraction expansion are bounded:

\begin{theorem}{Mushrooms with bounded partial quotients.}\label{th:finite}
Consider a MUPO-free mushroom for which $\rho$ has bounded partial quotients.  Then 
\begin{equation}
\frac{\limsup t^2P(t)}{\liminf t^2P(t)}\geq\frac{32}{27}\approx 1.185\ldots
\end{equation}
as $t\to\infty$, and in particular both limits are positive and finite.
\end{theorem}

The main idea of the proof is to show that the graph of $P(t)$ is approximately piecewise linear with the ratio between nonsmooth abscissas
at least the golden ratio infinitely often, and that a linear piece with this ratio has $t^2P(t)$ varying by at least a factor of $32/27$.  In
the special case where the partial quotients are not only bounded but eventually periodic we have

\begin{theorem}{Mushrooms with eventually periodic partial quotients.}\label{th:per}
Consider a MUPO-free mushroom for which $\rho$ has eventually periodic partial quotients.  Then there is a constant $\beta>1$
so that the limit taken over integers $m$,
\begin{equation}
\lim_{m\to\infty} \beta^{2m}t^2 P(\beta^mt)
\end{equation}
converges for all $t>0$.
\end{theorem}

The even partial quotients of $\rho$ must be bounded for the mushroom to be MUPO-free, from previous studies (see Sec.~\ref{s:prev} below).
However there is no such constraint on the odd partial quotients; typical MUPO-free mushrooms (if a measure on the set of MUPO-free
$\rho$ can be naturally defined) would be expected to have unbounded odd $a_n$, and indeed these values could grow rapidly.
It would be interesting to extend the methods developed here to these cases.  For now, we present only the example
in Fig.~\ref{f:mb37}, also with a sample size of $10^{11}$.

\begin{figure}
\centerline{\includegraphics[width=400pt]{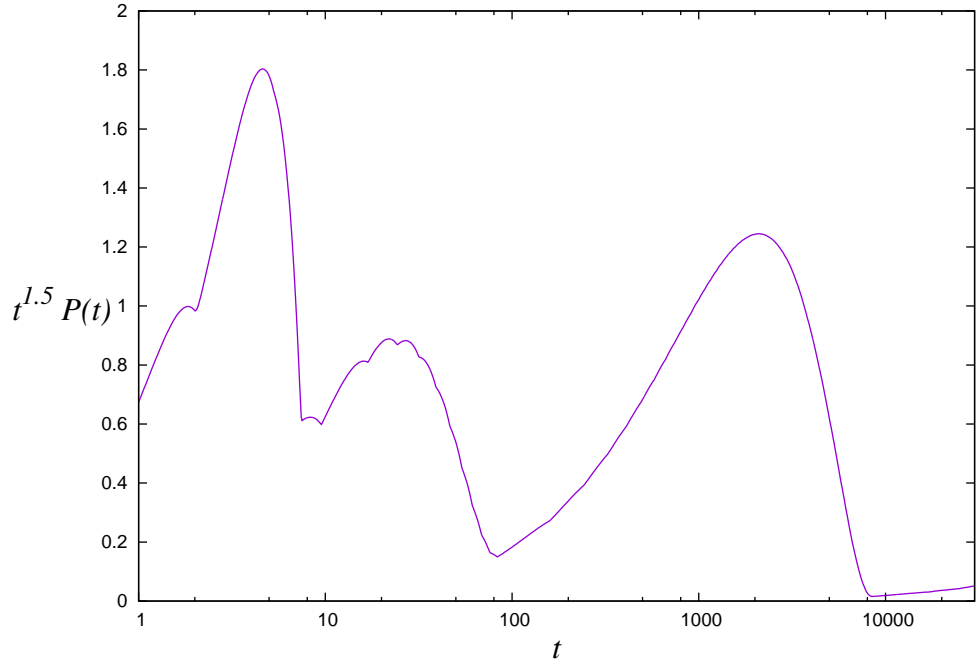}}
\vspace{-30pt}
\caption{A MUPO-free example constructed from the well-approximable $\rho=[0;1,3,10,1,100,1,1000,\ldots]$ which gives
$r\approx0.374600$.  An exponent for the survival probability $-\lim_{t\to\infty}\frac{\ln P(t)}{\ln t}$ does not appear to exist.}
\label{f:mb37}
\end{figure}

It appears that such behaviour may be arbitrarily close to that of MUPOs (see also Thm 1 in Ref.~\cite{D13} for a similar result in open chaotic
maps):
\begin{conjecture}{Liouville mushrooms.}
For sufficiently well-approximable MUPO-free $\rho$,
\begin{equation}
\limsup_{t\to\infty}\frac{-\ln P(t)}{\ln t}=1
\end{equation}
\end{conjecture}

There are some other systems to which the present results could apply.
Mushrooms are of particular interest in quantum mechanics, in which the classical dynamics approximates the
small wavelength limit of the Schr\"odinger equation with (typically) Dirichlet conditions on the boundary.
Here, wave functions corresponding to different components of phase space, and tunnelling rates
between these components, can be observed numerically and
experimentally~\cite{BKLRVHKS08,BB07,DFMRS07,Gomes15,YWHL14}.
Thus, it would be of interest to know in what way the amount of classical stickiness affects the quantum mechanical
properties such as the energy level spacing.

A closely related class of billiards, also used in experiments, is that of the annulus, a circle with
a circular scatterer~\cite{DGHHRR00}.  In this case orbits with sufficiently large angular momentum
cannot reach the scatterer and form a regular component of phase space.  Orbits reaching an off-centre
scatterer may be chaotic or belong to elliptic islands.  It would be interesting to study boundary stickiness
in both classical and quantum annuli.

Finally, we return to generic area preserving maps.  Given the subtleties for the case of sharply divided phase space,
it is unsurprising that the generic mixed boundary has eluded detailed understanding for so long.  There are similarities between
the mushroom and generic cases, for example the importance of one-sided Diophantine approximation; see Ref.~\cite{AFM15}
for a recent study focusing on the H\'enon map, and references.  There also, the rotation number of the
invariant boundary circle depends intricately on the control parameter.  This suggests that there may be an exceptional
set of parameters in which the non-integer universal decay exponent as conjectured in Ref.~\cite{CK08} does not hold.

The outline of this paper is as follows: Section 2 summarises relevant previous results, section 3 develops the theory, culminating
in proofs of the two theorems in section 4, with more technical lemmas and their proofs relegated to section 5.  Sec.~\ref{s:incmupo}
demonstrates the possibility of an incipient MUPO not leading to stickiness, but contains a further conjecture that they are not found
in otherwise MUPO-free mushrooms.

\section*{Acknowledgements}
This paper is dedicated to the memory of Nikolai Chernov, whose life continues to inspire and enlighten all who play mathematical billiards.
Thanks to the organisers of his memorial conference in Birmingham AL, May 2015 for their kind hospitality.
The author is grateful to Alastair Robertson, whose undergraduate project report~\cite{Robertson13} included an early version of
Fig.~\ref{f:mb64}, providing the impetus for this work; also to Leonid Bunimovich and Orestis Georgiou for helpful discussions and to
user O.L. on math.stackexchange.org for assistance with the final integral in Eq.~(\ref{e:arccos}).  This work was supported by the
EPSRC [grant number EP/N002458/1].  Data underlying this work (used to illustrate, not assist, the proofs)
is available from the University of Bristol data repository at https://dx.doi.org/10.5523/bris.ujws3ienvh8q1arcsm7c4q7ew

\section{Previous results}\label{s:prev}
This section gives a brief summary of previous results on MUPOs in mushrooms, or lack thereof.
In mushrooms with circular caps, each MUPO (defined by coprime integers $1\leq A<B$) exists for a fixed interval of stem radii~\cite{AMK05,DG11}
\begin{equation}
\cos\left(\frac{\pi}{2}\frac{A}{B}\right)\leq r < \frac{\cos\left(\frac{\pi}{2}\frac{A}{B}\right)}{\cos\left(\frac{\pi}{2}\frac{1}{B}\right)}
\end{equation}
For a fixed
mushroom, the amount and type of stickiness depends on what MUPOs exist.  This investigation was initiated
by Altmann and coauthors~\cite{Altmann,AMK05,AFMKR08}, who showed that for both mushroom and annular billiards,
the MUPOs are related to the Diophantine properties of the relevant parameter, here $\rho$.  MUPOs appear when there is sufficiently
good one-sided approximation of the component boundary by orbits with rational rotation numbers,
for example corresponding to unbounded even partial quotients in the continued
fraction expansion of $\rho$.  Such behaviour is typical, so that a full measure of parameters have infinitely
many MUPOs.  Rational $\rho$ yield a positive finite number of MUPOs, including the one
exactly on the boundary.  Tsugawa and Aizawa have recently studied the Fibonacci
case, that is, the most extreme badly approximable $\rho$~\cite{TA14}.

Further results were provided by Dettmann and Georgiou~\cite{DG11}, who showed that
there were parameter values with no MUPOs, gave a method for finding them, and the explicit example
$r=\cos\left(\frac{5+\sqrt{2}}{23}\pi\right)\approx 0.640134$. The idea here is that quadratic irrationals
have periodic (and in particular bounded) continued fraction expansions; a finite number of other conditions
needed to be checked.  Periodicity is not necessary; any value with bounded even partial quotients and sufficiently
small stem, specifically $a_{max}=\sup a_{2n}$ and
\begin{equation}
r<\frac{1}{\sqrt{\pi^2(a_{max}+2)^2+1}}
\end{equation}
will have finitely many MUPOs, and usually zero (can be proved for specific $r$ by checking a finite number of conditions) as in the
example in Fig.~\ref{f:mb37}.

Ref.~\cite{DG11} also showed that the supremum of
MUPO-free radii is $\frac{1}{\sqrt{2}}\approx 0.707107$, and that the supremum of finitely many MUPOs for
irrational rotation number is $\frac{4}{\sqrt{16+\pi^2}}\approx 0.786439$.  Bunimovich has recently
provided an alternative characterisation of the MUPO-free parameter set, including a number of theorems,
bounds on MUPO numbers and method for finding MUPO-free mushrooms~\cite{Bunimovich14}.  

\section{Development of the theory}
\subsection{From the mushroom to circle rotations}
The mushroom geometry we consider is shown on the left in Fig.~\ref{f:mush}.  The cap has radius $1$ while the stem has
radius $r\in(0,1)$ and has a polygonal geometry that contains no periodic orbits entirely within it.  An example with periodic
orbits would be a rectangle, with horizontal period two orbits.  Any such periodic orbits would be MUPOs and lead to
stickiness in the chaotic region, since all orbits of polygonal billiards are parabolic.  Apart from the no-MUPO constraint,
the stem geometry is arbitrary, and not relevant to the analysis below.  Note that many other mushroom geometries
can be constructed with different and interesting properties~\cite{BV12}.

Using the usual reflection trick, any orbit which hits only the cap is equivalent to an orbit in the circle obtained by reflecting
the cap in its straight sides (right of Fig.~\ref{f:mush}).  The stem then becomes a slit in the interior of the circle.
A MUPO is any periodic orbit which remains in the cap, but which intersects the circle of radius $r$ shown on the left of
Fig.~\ref{f:mush}.  A MUPO orbit may be rotated around the circle (showing that these come in continuous families of orbits),
however a small perturbation which changes its rotation number causes it to precess, and eventually reach the stem,
as demonstrated for the orbit shown, which is near a period three MUPO. 

We then use the 2-fold rotational symmetry of the circle and hole, and identify opposite points.  The collision map then
corresponds to a circle rotation
\begin{equation}\label{e:rot}
\Phi(x)= \{x+\lambda\}
\end{equation}
where $\pi x$ is arc length and $\{\}$ denotes fractional part (ie mod 1).  The slit corresponds to a hole (single due to the symmetry reduction) in the $x$
dynamics of size
\begin{equation}\label{e:h}
h=\frac{2}{\pi}\arccos\frac{l}{r}=\sqrt{\frac{4}{\pi}\tan\frac{\pi\rho}{2}}(\lambda-\rho)^{1/2}+\mathcal{O}(\lambda-\rho)^{3/2}
\end{equation}
where (recalling Eqs.~\ref{e:lam},~\ref{e:rho})
\begin{eqnarray}
l=\cos\frac{\pi\lambda}{2}\\
r=\cos\frac{\pi\rho}{2}
\end{eqnarray}
We consider only the chaotic part of the mushroom cap $\rho<\lambda<1$.  In continuous time, the collisions occur at intervals
\begin{equation}
\tau_\lambda=2\sqrt{1-l^2}=2\sin\frac{\pi\lambda}{2}
\end{equation}

Here, we are interested in mushrooms without MUPOs.  As shown in Refs.~\cite{Altmann,DG11,Bunimovich14} this
corresponds to a set of $r$ with zero Lebesgue measure with a supremum at $r=1/\sqrt{2}$, described using the
Diophantine properties of $\rho$.

\subsection{Continued fractions and the three gap theorem}\label{s:3gt}
Dynamics in a circle is described by the three gap theorem~\cite{Slater67,Sos58,Ravenstein88}; that is, the set
$\{\Phi^i(x)\}_{i=0\ldots N-1}$ has gaps between adjacent points on the circle of at most
three different sizes, which are determined by $N$ and the continued fraction expansion of
$\lambda$.  We follow the first chapter of Ref.~\cite{RS92}, extending the notation to the
``semiconvergents'': Given a standard continued
fraction for $0<\lambda=a_0+1/(a_1+1/(a_2+\ldots))=[a_0;a_1a_2\ldots]$ containing
partial quotients $a_j\in\mathbb{Z}$, with $a_0\geq 0$ (here $a_0=0$) and all other
$a_j\geq 1$ we define $B_{-1}=0$, $A_{-1}=B_0=1$, $A_0=a_0$, $A_k=A_{k,a_k}$
and $B_k=B_{k,a_k}$ with
\begin{eqnarray}
A_{k,i}&=&iA_{k-1}+A_{k-2}\\
B_{k,i}&=&iB_{k-1}+B_{k-2}\label{e:Brec}
\end{eqnarray}
for $k\geq 1$ and $1\leq i\leq a_k$.  Then $A_k/B_k$ are the convergents (closest approximants)
to $\lambda$ and $A_{k,i}/B_{k,i}=[a_0;a_1a_2\ldots a_{k-1}i]$ are semiconvergents.  
The complete quotients are defined as $\zeta_k=[a_{k};a_{k+1}a_{k+2}\ldots]$ and hence lie
between $a_k$ and $a_k+1$.  In terms of these we have for any $k\geq 0$,
\begin{equation}
\lambda=\frac{A_k\zeta_{k+1}+A_{k-1}}{B_k\zeta_{k+1}+B_{k-1}}
\end{equation}
The differences are $D_k=D_{k,a_k}$ with
\begin{eqnarray}
D_{k,i}&=&B_{k,i}\lambda-A_{k,i}\nonumber\\
&=&iD_{k-1}+D_{k-2}\label{e:diff}\\
&=&(-1)^k\frac{(a_k-i)\zeta_{k+1}+1}{B_k\zeta_{k+1}+B_{k-1}}\nonumber
\end{eqnarray}
using the relation $A_kB_{k-1}-A_{k-1}B_k=(-1)^{k+1}$.  Thus we can determine the sign
and bound the magnitude of the differences.

These differences are plotted in Fig.~\ref{f:D}, a union of straight line segments labelled by coprime postive integers
$(A,B)$. Each has $x$-intercept $A/B$ and gradient $B$.  The final partial quotient of a continued
fraction expansion of $A/B$ is $a_k=i$.  Each straight line segment corresponds to two $(k,i)$, even $k$ for $D>0$,
odd $k$ for $D<0$, and $i=1$ for the larger $k$, corresponding to the two continued fraction representations of
the rational $A/B$.  Each segment extends to values of $\lambda$ equal to the truncated continued fractions
(Stern-Brocot parents) of $A/B$.  

\begin{figure}
\centerline{\includegraphics[width=400pt]{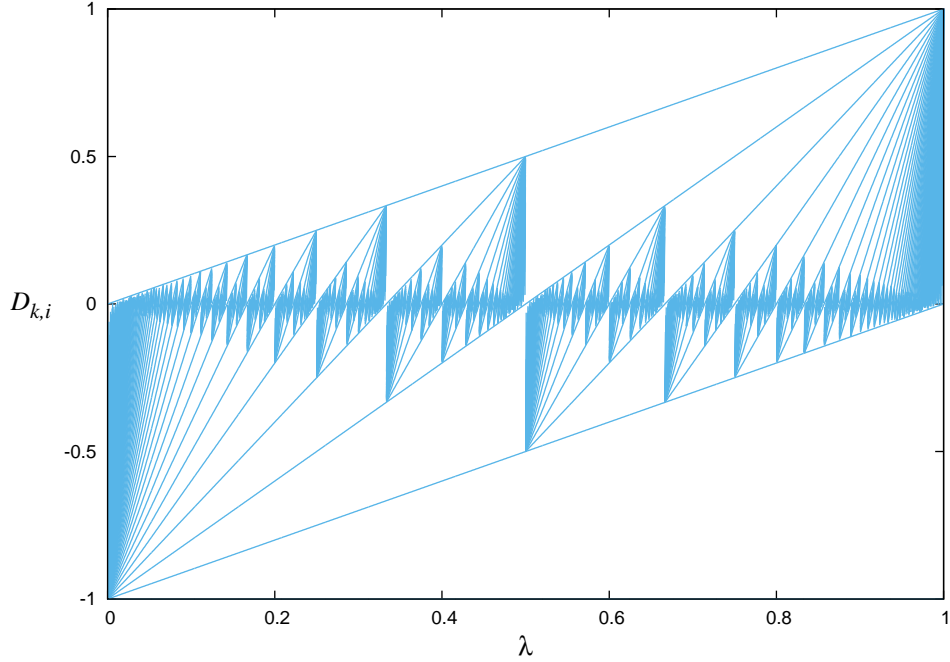}}
\caption{The differences $D_{k,i}$ as a function of $\lambda$. }
\label{f:D}
\end{figure}

In our notation, the three gap theorem says that we need to find the largest $B_{k,i}\leq N$.
Then there are $N-B_{k-1}$ gaps of size $|D_{k-1}|$, $N-B_{k,i}$ gaps of size $|D_{k,i}|$
and the remaining $B_{k-1}+B_{k,i}-N$ (possibly zero) gaps are the sum of the previous two
sizes, which comes to $|D_{k,i-1}|$ ($=|D_{k-2}|$ if $i=1$).

\subsection{Variation of $\lambda$}\label{s:incmupo}
The three gap theorem describes the gaps at fixed $\lambda$, while the initial conditions
for the survival probability are at all $\lambda>\rho$.  So, we need to describe how the variation
in $\lambda$ affects the relevant part of the continued fraction expansion, and the length of
the relevant part, namely the $(k,i)$ that we need to consider.

Consider the dynamics at fixed $\lambda$.
Initially, all gaps are larger than $h$.  Best approximation is found for the
full convergents, so the first difference to decrease below $h$ must be from a convergent,
say $|D_{k-1}|$ with $k$ now determined by $h$.  The second will be in the next sequence
of semiconvergents, $|D_{k,i}|$, possibly the next convergent (if $i=a_k$).  Now, only the
largest of the three gaps, namely $|D_{k,i-1}|$ remains greater than $h$.  Finally, the next
difference, which is the difference of these two, $|D_{k,i+1}|$ or $|D_{k+1,1}|$ will also fall
below $h$, leading to complete escape.  Thus we have defined a specific $(k,i)$ for each
$\lambda$ in a given mushroom (parametrised by $\rho$).

Now, the condition for a MUPO can be expressed simply in terms of these quantities:
If there is a $\lambda=A/B$ such that $0\leq h<1/B$ the orbit is periodic and can avoid the hole,
so we have a MUPO.  In this case, all the differences $D_{k,i}$ are multiples of $1/B$ and so
greater than $h$.  Conversely, if there is at least one nonzero difference less than $h$ at all
$\lambda\geq \rho$, then no MUPO exists.  See Fig.~\ref{f:Dh}.

\begin{figure}
\centerline{\includegraphics[width=400pt]{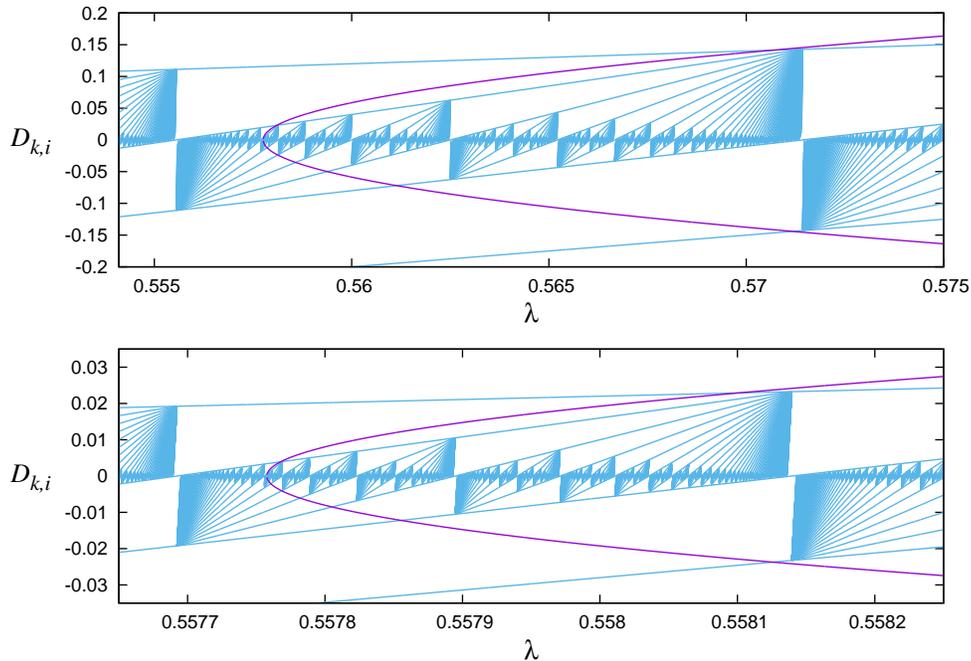}}
\caption{Close-ups of Fig.~\protect\ref{f:D} showing $\pm h$ (the hole size) for the known MUPO-free case
$r=\cos\left(\frac{5+\sqrt{2}}{23}\pi\right)$ as the parabola-like curve.  Note that this value of $r$ barely
avoids a MUPO at $\lambda=4/7=0.571\ldots $, since $h$ is only slightly greater than $1/7$ at that point. Also,
due to the periodicity of the continued fraction expansion, the behaviour near the intersection
with the $x$-axis is self-similar.}
\label{f:Dh}
\end{figure}

Transitions in the continued fraction expansion itself take place at rational $\lambda$.  However,
there is at least one non-zero difference less than $h$, which we can identify as $D_{k-1}$.  The next
to fall below $h$, $D_{k,i}$ is a negative multiple of this at the rational point, so also non-zero.  Thus
none of the quantities needed to define these differences, $k$, $i$ or $a_{k-1}$ have changed across
the transition.  Furthermore, as shown in Lemma~\ref{l:mon} the value of $k$ can only decrease
as $\lambda$ increases.  Together, these show that variations in the continued fraction expansion in
$\lambda$ are never relevant; the survival probability $P(t)$ may be calculated using the (fixed)
expansion of $\rho$.

Transitions in $(k,i)$ as $\lambda$ varies must therefore correspond only to cases where
$|D_{k,i}|=h$ and/or $|D_{k-1}|=h$.  Note that
\begin{itemize}
\item $D_{k,i}$ and $h$ increase with $\lambda$.
\item $D_{k,i}$ is linear and $h$ concave in $\lambda$.  
\item $D_{k,i}>0$ iff $k$ is even (see Eq.~\ref{e:diff}).
\end{itemize}

The second point shows that $|D_{k,i}|=h$ can have at most two solutions.  Actually, as shown in
Lemma~\ref{l:mon}, only the solution where $h$ is increasing faster than $D_{k,i}$ is relevant,
since the other possibility leads to the existence of a MUPO.

\begin{definition}{Incipient MUPO.}
Transition of both types occurring simultaneously, that is, $|D_{k,i}|=|D_{k-1}|=h$.
\end{definition}
At the relevant $\lambda$, the hole size is exactly the same size as both relevant differences.  Increasing
$\rho$ by an arbitrarily small amount leads to a MUPO, hence the terminology.  At the transition point
itself, all orbits
at and near this value of $\lambda$ escape in finite time, so there is no effect on the long time
survival probability.  In order to satisfy the equation at some rational rotation number $A/B$ , we have
\begin{equation}
r=\frac{\cos\frac{\pi}{2}\frac{A}{B}}{\cos\frac{\pi}{2}\frac{1}{B}}
\end{equation}

\begin{conjecture}{No incipient MUPOs without MUPOs.}
The above equation for $r$ is never satisfied without the existence of another (real) MUPO.
\end{conjecture}

The equation for an incipient MUPO has a countable number of solutions, while the condition for the 
non-existence of a MUPO has (from previous studies) a zero measure set of solutions.  Thus, a
probabilistic argument suggests there is no overlap.  In particular, there is no reason to expect the continued
fraction expansion of $\frac{2}{\pi}\arccos r$ to have bounded even partial quotients, which is a
necessary condition for being MUPO-free.  A numerical search does not find any solutions,
although there are examples for which the smallest MUPO is rather long.   One of the simplest is
$\lambda=\frac34$, $h=\frac14$, $l=\cos \frac{3\pi}{8}=\frac{\sqrt{2-\sqrt{2}}}{2}$, $r=\frac{l}{\cos(\pi/8)}=\sqrt{2}-1$.  In
this case, the shortest MUPO has $A/B=1181/1622$.

\subsection{Survival probability at fixed $\lambda$}
Now let us calculate the survival probability, that is, the measure of surviving orbits in the
presence of the hole of size $h$.  For now, the initial conditions have fixed $l$ and hence
$\lambda$ and $h$ but are otherwise uniformly distributed round the circle. Recall that
the condition for a MUPO is that there is a $\lambda=A/B$ for which $h<1/B$.

Putting the above pieces together, we find the survival probability as a function of collisions:
\begin{equation}\label{e:P(N)}
P_\lambda(N)=\left\{\begin{array}{cc}
1-Nh&N\leq B_{k-1}\\
1-B_{k-1}h-(N-B_{k-1})|D_{k-1}|&B_{k-1}<N\leq B_{k,i}\\
(B_{k-1}+B_{k,i}-N)(|D_{k,i-1}|-h)&B_{k,i}<N\leq B_{k-1}+B_{k,i}\\
0&N> B_{k-1}+B_{k,i}
\end{array}\right.
\end{equation}
Note that for moderate $h$, it is quite possible for the most of the particles to remain for a long time if $a_k$ is large and hence $i$ and $B_{k,i}$
can be large.

\subsection{Integrating over $\lambda$}
We place initial conditions uniformly in the semicircular cap of the mushroom with uniform directions; this corresponds to the equilibrium invariant measure of the
billiard flow: See for example Ref.~\cite{CM06}.  Integrating over an arbitrary function of $\lambda$, we have
\begin{eqnarray}
\int f(\lambda)d\mu&=&\int_0^{2\pi}\frac{d\phi}{2\pi}\int_0^\pi \frac{d\theta}{\pi}\int_0^1 2r dr f\left(\frac{2}{\pi}\arccos|r\sin(\theta-\phi)|\right)\nonumber\\
&=&\int_0^1 f(\lambda)2\sin^2\left(\frac{\pi}{2}\lambda\right)d\lambda\\
&=&\int f(\lambda)d\mu_\lambda\nonumber
\end{eqnarray}
giving explicitly the associated measure $d\mu_\lambda$ for $\lambda\in (0,1)\subset \mathbb{R}$. Here, $r$ is radial distance, $\theta$ gives the angular position, $\phi$ gives the
direction of the particle, hence the angular momentum is $l=|r\sin(\theta-\phi)|$.

We further restrict to the chaotic region, $l<r$ or equivalently $\lambda>\rho$, which requires normalisation by a further constant
\begin{equation}
c_\rho=\int_\rho^1 2\sin^2\left(\frac{\pi}{2}\lambda\right)d\lambda=1-\rho+\frac{2}{\pi}\sin\frac{\pi \rho}{2}\cos\frac{\pi\rho}{2}
\end{equation}

Thus, the full survival probability is
\begin{equation}
P(t)=c_\rho^{-1}\int  P_\lambda\left(\lfloor\frac{t}{\tau_\lambda}\rfloor\right)d\mu_\lambda
\end{equation}
which in principle can be evaluated exactly at fixed $t$, splitting $\lambda$ into regions with differing $N=\lfloor t/\tau_\lambda \rfloor$ (if applicable), and into
the different regimes of Eq.~(\ref{e:P(N)}), and noting the integrals
\begin{eqnarray}
\int d\mu_\lambda&=&\lambda-\frac{2}{\pi}l\sqrt{1-l^2}+C\nonumber\\
\int hd\mu_\lambda&=&-\int\left(\frac{2}{\pi}\arccos\frac{l}{r}\right)\left(\frac{4}{\pi}\sqrt{1-l^2}\right)dl\nonumber\\
&=&\frac{2}{\pi^2}\left[{\rm Li}_2\left(r^{-1}e^{i\pi(\lambda+h)/2}\right)+{\rm Li}_2\left(r^{-1}e^{-i\pi(\lambda+h)/2}\right)\right.\nonumber\\
&&-\frac{\pi^2}{4}(h^2-2h-2\lambda)-\pi lh\sqrt{1-l^2}\label{e:arccos}\\
&&\left.+(1-r^2)\ln\left(\sqrt{1-l^2}+\sqrt{r^2-l^2}\right)+\sqrt{1-l^2}\sqrt{r^2-l^2}\right]+C\nonumber
\end{eqnarray}
where ${\rm Li}_2$ is the dilogarithm and $h$ and $l$ are the usual functions of $\lambda$.  We do not need the explicit form of $P(t)$ for any of our
results, however. 

\section{Proofs of the theorems}
Sec.~\ref{s:bounds} shows that the positive finite limits in Thm.~\ref{th:finite}.  Sec.~\ref{s:pwl} uses the almost piecewise linearity of $P(t)$ to get the explicit bound,
completing the proof of Thm.~\ref{th:finite}. Sec.~\ref{s:per} contains the proof of Thm.~\ref{th:per}.

\subsection{Bounds on the survival probability}\label{s:bounds}
From this point we assume that in addition to the MUPO-free condition, the partial quotients of $\rho$ are bounded.  Possible effects of violating this condition were
discussed briefly in the introduction.  We have Lemma~\ref{l:interval} which considers the interval of relevant $\lambda$ for large $t$.  In particular, there is
a function $\lambda_{max}(t)$ satisfying
\begin{equation}\label{e:t2}
\lambda_{max}(t)=\rho+\mathcal{O}(t^{-2})
\end{equation}
at large $t$ so that $P_\lambda\left(\lfloor\frac{t}{\tau_\lambda}\rfloor\right)=0$ for $\lambda>\lambda_{max}(t)$.  

The integral for $P(t)$ directly gives an upper bound
\begin{equation}
P(t)<\frac{C_{max}}{t^2}
\end{equation}
To get the equivalent lower bound
\begin{equation}
P(t)>\frac{C_{min}}{t^2}
\end{equation}
at sufficiently large $t$ we note that $P_\lambda(N)\geq 1-Nh$ for all $N$ and integrate until this vanishes, a region of order $t^{-2}$.
This completes the proof of positive and finite limits in Theorem~\ref{th:finite}.

\subsection{Approximate piecewise linearity}\label{s:pwl}
We see from Eq.~\ref{e:P(N)} that $P_\lambda(N)$ is a piecewise linear function of $N$.  We would like to see whether this applies also to $P(t)$, which we now
approximate.  It is easy to see that Eq.~\ref{e:t2} implies
\begin{equation}
\frac{t}{\tau_\rho}-\frac{t}{\tau_{\lambda_{max}}}=\mathcal{O}(t^{-1})
\end{equation}
that is, there are at most two relevant values of $N$ at sufficiently long times.  From Eq.~\ref{e:P(N)}, $0<P_\lambda(N)< 1$ and $P_\lambda(N+1)-P_\lambda(N)\leq h$,
where $h=\mathcal{O}(t^{-1})$. Putting this together we see that
\begin{equation}
|P(t)-\tilde{P}(t)|=\mathcal{O}(t^{-3})
\end{equation}
where the approximated integral
\begin{equation}
\tilde{P}(t)=c_\rho^{-1}\int P_\rho\left(\lfloor \frac{t}{\tau_\rho}\rfloor\right) d\mu_\lambda
\end{equation}
is piecewise linear with transitions at $\tau_\rho B_{k,i}$.  This is observed numerically in Fig.~\ref{f:pwl}.

\begin{figure}
\centerline{\includegraphics[width=400pt]{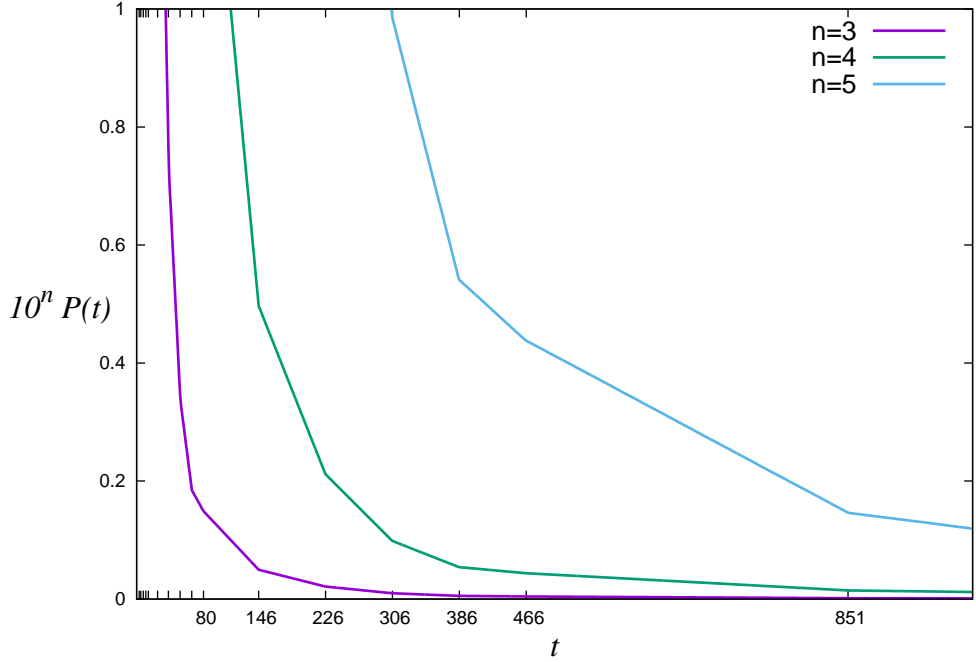}}
\caption{The survival probability is approximately piecewise linear, with pieces defined by $\tau_\rho B_{k,i}$,
here represented on the axis by the nearest integer.  As above, we choose
$r=\cos\left(\frac{5+\sqrt{2}}{23}\pi\right)$.}
\label{f:pwl}
\end{figure}

Furthermore, infinitely many transition points are spaced with ratios greater than the golden ratio, as shown in Lemma~\ref{l:bratg}, and a piecewise linear
function such as $\tilde{P}(t)$ with asymptotic form $t^{-2}$ and transition points with such spacing must have a limiting ratio of at least $\frac{32}{27}$ as
shown in Lemma~\ref{l:32/27}. From above, $|\tilde{P}(t)-P(t)|$ is at most of order $t^{-3}$ and so negligible in the limit, so the result holds for $P(t)$ as well.
This complete the proof of Theorem~\ref{th:finite}.

\subsection{Eventually periodic continued fractions}\label{s:per}
In this final section we prove Theorem~\ref{th:per}, for which the continued fraction of $\rho$ is eventually periodic.
As shown in standard texts including Ref.~\cite{RS92}, this condition is exactly that $\rho$ is a quadratic irrational.

We can write the recurrence relation Eq.~\ref{e:Brec} as
\begin{equation}
v_k=A_kv_{k-1}
\end{equation}
with
\begin{eqnarray}
v_k&=&\left(\begin{array}{c}B_{k}\\B_{k-1}\end{array}\right)\\
A_k&=&\left(\begin{array}{cc}a_k&1\\1&0\end{array}\right)
\end{eqnarray}
so that $\det A_k=-1$. Also, any product of $A$ matrices has only positive entries, so by the Perron-Frobenius theorem it has a simple positive eigenvalue of strictly
maximal magnitude.  Let $p$ be the period of the continued fraction expansion of $\rho$.  All products of $2p$ consecutive $A$ matrices in the periodic part of
the expansion have unit determinant.  They are also cyclic permutations and hence have the same eigenvalues $\beta>\beta^{-1}>0$.
Since $\beta$ and $\beta^{-1}$ are the roots of a monic quadratic polynomial with integer coefficients they are conjugate quadratic irrationals.

Now let $k$ be in the periodic part of the expansion (that is, sufficiently large, and fixed).  We have
\begin{equation}
B_k=C_k\beta^\frac{k}{2p}+\mathcal{O}(\beta^{-\frac{k}{2p}})
\end{equation}
where $C_{k+2p}=C_k$.  Noting Eq.~\ref{e:diff} and the periodicity of $\zeta_k$, we find
\begin{equation}
D_k=\tilde{C}_k\beta^{-\frac{k}{2p}}+\mathcal{O}(\beta^{-3\frac{k}{2p}})
\end{equation}
where $\tilde{C}_{k+2p}=\tilde{C}_k$, and equations with related periodic constants for $B_{k,i}$ and $D_{k,i}$.  This self-similarity of the differences can be observed in
Fig.~\ref{f:Dh} above.

We now define $t'=\beta^mt$ for some large integer $m$ and large $t$, and want to compare $\tilde{P}(t')$ with $\tilde{P}(t)$.  We see that the various quantities scale: $N'=t'/\tau_\rho=\beta^mN$. For most $t$, we have $(k',i')=(k+2mp,i)$, so we consider $D_{k+2pm,i}=\beta^{-m}D_{k,i}(1+\mathcal{O}(\beta^{-\frac{k}{p}}))$,
equating this with $h$ to find the transitions in $\lambda$, we have, using Eq.~\ref{e:h}, $\lambda'-\rho=\beta^{-2m}(\lambda-\rho)(1+\mathcal{O}(\beta^{-\frac{k}{p}}))$.  Thus
when $(k',i')=(k+2mp,i)$,
\begin{equation}\label{e:Ptilde}
|\beta^{2m}t^2\tilde{P}(\beta^mt)-t^2\tilde{P}(t)|=\mathcal{O}(\beta^{-\frac{2k}{p}})
\end{equation}
The nonsmooth points of $\tilde{P}(t)$ do not exactly scale: From above we have $\tau_\rho B_{k+2mp,i}=\beta^{m}\tau_\rho B_{k,i}(1+\mathcal{O}(\beta^{-\frac{k}{p}}))$.
This means that for $t$ close to a transition point, it is possible that $(k',i')\neq (k+2mp,i)$.  In this case we consider a value $\tilde{t}$ close to $t$ but across the transition.
We have $|\tilde{t}-t|=\mathcal{O}(\beta^{-\frac{k}{2p}})$.  Also, $t^2\tilde P(t)$ has positive upper and lower bounds and $\tilde{P}(t)$ is convex, implying that the variation satisfies
\begin{equation}
\left|\frac{\tilde{P}(\tilde{t})-\tilde{P}(t)}{\tilde{t}-t}\right|<\frac{C}{\min(\tilde{t},t)^3}
\end{equation}
so that Eq.~\ref{e:Ptilde} is still satisfied.  Finally we use $\tilde{P}(t)$ to approximate $P(t)$ and take $m\to\infty$ then $k\to\infty$ to obtain the result of Theorem~\ref{th:per}.

\section{Lemmas and their proofs}
\begin{lemma}{Monotonicity of transitions.}\label{l:mon}
A difference becomes relevant only as $\lambda$ increases.  Precisely, for a MUPO-free $r$, if there is some $(k,i)$ and $\lambda$ for which
$f(\lambda)=h-|D_{k,i}|=0$, then $f'(\lambda)>0$ at that point.
\end{lemma}

\begin{proof}
If $k$ is odd, then $D_{k,i}<0$, so both terms in $f'(\lambda)$ are positive and we are finished.

In the case $k$ is even, assume that the assertion is false.  Then $D_{k,i}$ continues to higher $\lambda$ until it terminates at $\lambda=A_{k-1}/B_{k-1}$
at which it is equal to $1/B_{k-1}$ using the equations in Sec.~\ref{s:3gt}.
Since $D_{k,i}$ is a linear function of $\lambda$ and $h$ is concave, $f'(\lambda)$ remains negative, hence $f(A_{k-1}/B_{k-1})<0$.
This means that $0<h<1/B_{k-1}$ at this point, which would imply that there is a MUPO with rotation number $A_{k-1}/B_{k-1}$, a contradiction.
\end{proof}

\begin{lemma}{Size of $\lambda$ interval.}\label{l:interval}
For MUPO-free $\rho$ with bounded partial quotients, $P_\lambda\left(\lfloor\frac{t}{\tau_\lambda}\rfloor\right)=0$ for $\lambda>\lambda_{max}(t)$ where
\begin{equation}
\lambda_{max}(t)=\rho+\mathcal{O}(t^{-2})
\end{equation}
as $t\to\infty$.
\end{lemma}

\begin{proof}
Fix a time $t>0$ and constant $0<C_1<t^2(1-\rho)$.  If
\begin{equation}
\lambda>\rho+\frac{C_1}{t^2}
\end{equation}
we have
\begin{equation}
h>\frac{C_2}{t}
\end{equation}
using the expansion of $h$, Eq.~\ref{e:h}.  Here, $C_2$ is a constant proportional to $C_1$.
The value of $k$ is determined so that $|D_{k-1}|$ is the first difference to fall below $h$.  In particular
\begin{equation}
|D_{k-2}|>h
\end{equation}
Equation 1.4.5 of Ref~\cite{RS92} gives
\begin{equation}
|D_{k-2}|\leq\frac{1}{B_{k-1}}
\end{equation}
Thus we have
\begin{equation}
B_{k-1}<\frac{t}{C_2}
\end{equation}
Using the recurrence relation for the $B_k$, the boundedness of the $a_k$ (because there are no MUPOs
we may use the continued fraction expansion of $\rho$ rather than $\lambda$) we have 
\begin{equation}
B_{k-1}+B_{k,i}<\frac{t}{C_3}
\end{equation}
with a constant $C_3$ proportional to $C_2$.  We also have
\begin{equation}
N=\lfloor \frac{t}{\tau_\lambda} \rfloor>\frac{t}{2}-1>\frac{t}{3}
\end{equation}
since $\tau_\lambda<2$ for all $\lambda$ and $t$ can be chosen arbitrarily large.  If we choose $t$ large enough, we can choose a large enough $C_1<t^2(1-\rho)$ so
that $C_3>3$ and we find from Eq.~(\ref{e:P(N)}) that $P_\lambda(N)=0$ as required.
\end{proof}

\begin{lemma}{Ratio of semiconvergent denominators.}
For any $\lambda\not\in\mathbb{Q}$, its (infinite) sequence of semiconvergent denominators $B_{k,i}$ has ratio of successive terms at least
$g=\frac{1+\sqrt{5}}{2}$ infinitely often.\label{l:bratg}
\end{lemma}

\begin{proof}
We have
\[ \frac{B_{k,1}}{B_{k-1}}=\frac{B_{k-1}+B_{k-2}}{B_{k-1}}=1+\frac{B_{k-2}}{B_{k-1}}=[1;a_{k-1},\ldots, a_1] \]
where the final equality comes from Ref.~\cite{RS92}, section 1.6.

If $\lambda\sim g$, that is, the $a_k$ have a tail consisting only of $1$s, we have $B_{k-2}/B_{k-1}\to g^{-1}$ alternating above and below this
value, and we are finished.

If there are are infinitely many $1$s in the partial quotients, any $k$ for which $a_{k-1}=1$ and $a_{k-2}>1$ will have $[1;a_{k-1},a_{k-2}\dots,a_1]>5/3>g$
and we are likewise finished.

If all but a finite number of partial quotients are greater than $1$, instead consider (for both $a_k>1$ and $a_{k-1}>1$)
\[ \frac{B_{k,2}}{B_{k,1}}=\frac{2B_{k-1}+B_{k-2}}{B_{k-1}+B_{k+2}}=1+\frac{1}{1+\frac{B_{k-2}}{B_{k-1}}}=[1;1,a_{k-1},a_{k-2},\ldots a_1]
>\frac{5}{3}>g \]
\end{proof}

\begin{lemma}{Minimum variation of piecewise linear functions.}\label{l:32/27}
Let $f:\mathbb{R}^+\to\mathbb{R}^+$ be piecewise linear with consecutive transition points $t_i$ having ratio at least the golden ratio infinitely often.  If the limits are positive and
finite as $t\to\infty$ then
\[ \frac{\limsup t^2f(t)}{\liminf t^2f(t)}\geq\frac{32}{27}\approx 1.185\ldots \]
\end{lemma}

\begin{proof}
According to the assumptions we can find infinitely many consecutive transition points (two of which denoted $t_1,t_2$) such that $t_2/t_1\geq g$.  Within such an interval, writing $t_2/t_1=G$,
\[ t^2f(t)=t^2\left[\beta-\gamma(t-t_1)\right] \]
where
\[ \beta=f(t_1)>0,\qquad \gamma=-\frac{f(G t_1)-f(t_1)}{(G-1)t_1}>0 \]
are constant.  Possible supremum/infimum points consist of the left and right endpoints, and a turning point:
\begin{eqnarray*}
L&=&t_1^2f(t_1)=\beta t_1^2\\
T&=&t_c^2f(t_c)=\beta t_1^2 \frac{4}{27z}(1+z)^3,\qquad \mbox{if }t_c=t_1\frac{2}{3}(1+z)\in(t_1,G t_1)\\
R&=&t_2^2f(t_2)=\beta t_1^2G^2\left[1-\frac{G-1}{z}\right]
\end{eqnarray*}
where $z=\beta/(\gamma t_1)$.
The turning point is relevant if
\[ \frac{1}{2}<z < \frac{3G}{2}-1 \]
First we turn to cases where it is not relevant.
For $z<1/2$ we find, using $G\geq g$, that $R<0$, which is impossible.  For $z>3G/2-1$ we compute
\[ \frac{R}{L}=G^2\left[1-\frac{G-1}{z}\right]>G^2\left[1-\frac{G-1}{3G/2-1}\right]=\frac{G^3}{3G-2} \]
The derivative of the right hand side is positive for $G\geq g$.  Thus we conclude
\[ \frac{R}{L}\geq \frac{g^3}{3g-2}>\frac{32}{27} \]
as required.  This completes the cases where the turning point is not relevant.

When the turning point is relevant, we have two cases.  When $z\geq\frac{G^2}{G+1}$ we have $R\geq L$, and so the relevant ratio is
\[ \frac{T}{L}=\frac{4}{27}\frac{(1+z)^3}{z}\geq \frac{32}{27} \]
since we know that $z\geq \frac{G^2}{G+1}\geq \frac{g^2}{g+1}=1$.
Conversely, when $z\leq \frac{G^2}{G+1}$, the relevant ratio is
\[ \frac{T}{R}=\frac{4}{27}\frac{(1+z)^3}{G^2(z+1-G)}\]
Differentiating (and keeping $G$ constant) we have
\[ \frac{d}{dz}\frac{T}{R}=\left(\frac{z+1}{G(z+1-G)}\right)^2\left(2(z+1)-3G\right)<0 \]
since $z<3G/2-1$ in order for $T$ to be relevant, as above.  Thus its minimum value is obtained at $z=\frac{G^2}{G+1}$:
\[ \frac{T}{R}\geq \frac{4}{27}\frac{(1+\frac{G^2}{G+1})^3}{G^2(\frac{G^2}{G+1}+1-G)}
=\frac{4}{27}\frac{(G^2+G+1)^3}{G^2(G+1)^2}
\]
This function is increasing for $G\geq g$, thus we find
\[ \frac{T}{R}\geq \frac{4}{27}\frac{(g^2+g+1)^3}{g^2(g+1)^2}=\frac{32}{27} \]
as required.
\end{proof}

\bibliographystyle{amsplain}
\bibliography{billiards}

\end{document}